\documentclass[journal]{IEEEtran}
\usepackage{graphics} 
\usepackage{epsfig} 
\usepackage{times} 
\usepackage{amsmath} \allowdisplaybreaks[4]
\usepackage{amsthm}
\usepackage{booktabs}
\usepackage{amssymb}
\usepackage[most]{tcolorbox}
\usepackage{multirow}
\usepackage{longtable}
\usepackage {threeparttable}
\usepackage{mathtools}
\usepackage[utf8]{inputenc}
\usepackage{tikz}
\usetikzlibrary{arrows.meta}
\usetikzlibrary{graphs}
\usepackage[colorlinks=false, urlcolor=blue, pdfborder={0 0 0}]{hyperref}
\usepackage{enumerate}
\usepackage{enumitem}
\usepackage{color}
\usepackage[linesnumbered,boxed,ruled,commentsnumbered]{algorithm2e}
\usepackage{subfigure}
\usepackage{comment}
\usepackage{cite}
\usepackage{flushend} \usepackage{float} 
\newtcolorbox{resp}[1][]{%
enhanced jigsaw,% 
size=small,%
boxrule=1pt,%
halign title=flush center,%
coltitle=black,%
breakable,%
drop shadow=black!50!white,%
attach boxed title to top left={xshift=1cm,yshift=-\tcboxedtitleheight/2,yshifttext=-\tcboxedtitleheight/2},%
minipage boxed title=3cm,%
boxed title style={%
	colback=white,%
	size=fbox,%
	boxrule=1pt,%
	boxsep=2pt,%
	underlay={%
		\coordinate (dotA) at ($(interior.west) + (-0.5pt,0)$);
		\coordinate (dotB) at ($(interior.east) + (0.5pt,0)$);
		\begin{scope}[gray!80!black]
			\fill (dotA) circle (2pt);
			\fill (dotB) circle (2pt);
		\end{scope}
	}%
},%
#1%
}
\newfloat{tcolorboxenv}{h}{lob}
\definecolor{myred}{RGB}{182,20,50}
\definecolor{myblue}{RGB}{227, 245, 250}
\definecolor{myyellow}{RGB}{255,255,0}
\definecolor{mygreen}{RGB}{250, 250, 235}
\definecolor{myorange}{RGB}{255,128,0}
\definecolor{mygray}{RGB}{192,192,192}
\newtheorem{mydef}{Definition}
\newtheorem{mythm}{Theorem}
\newtheorem{myprob}{Problem}
\newtheorem{mylem}{Lemma}

\newtheorem{remark}{Remark}
\newtheorem{assumption}{Assumption}
\usetikzlibrary{patterns}
\usetikzlibrary{shapes}

\definecolor{myblue}{RGB}{135,206,235}
\definecolor{myyellow}{RGB}{255,255,0}
\definecolor{mygreen}{RGB}{0,255,0}
\newfloat{tcolorboxenv}{h}{lob}

\title{On the Construction of Barrier Certificates: \\ A Dynamic Programming Perspective}

\author{Yu Chen, Shaoyuan Li and Xiang Yin% 
\thanks{This work was supported by  the National Natural Science Foundation of China (62061136004, 62173226, 61833012).}
	\thanks{Yu Chen, Shaoyuan Li and Xiang Yin are with School of Automation and Intelligent Sensing, Shanghai Jiao Tong University, Shanghai 200240, China.
	{\tt\small \{yuchen26, syli, yinxiang\}@sjtu.edu.cn}. Corresponding Author: Xiang Yin.}
}

\begin{document}
	
\maketitle

\begin{abstract}
In this paper,  we revisit the formal verification problem for stochastic dynamical systems over finite horizon using barrier certificates. 
Most existing work on this topic focuses on safety properties by constructing barrier certificates based on the notion of $c$-martingales. 
In this work, we first provide a new insight into the conditions of existing martingale-based barrier certificates from the perspective of dynamic programming operators. 
Specifically, we show that the existing conditions essentially provide a bound on the dynamic programming solution, which exactly characterizes the safety probability.  
Based on this new perspective, we demonstrate that the barrier conditions in existing approaches could be unnecessarily conservative over unsafe states. 
To address this, we propose a new set of safety barrier certificate conditions that are strictly less conservative than existing ones, thereby providing tighter probability bounds for safety verification.  
We further extend our approach to the case of reach-avoid specifications by providing a set of new barrier certificate conditions. 
We also illustrate how to search for these new barrier certificates using sum-of-squares (SOS) programming. 
Finally, we use two numerical examples to demonstrate the advantages of our method compared to existing approaches. 
\end{abstract}

\section{Introduction} 
Safety is a major concern in cyber-physical systems (CPS), especially as they become more integrated into critical infrastructures and develop complex functionalities. 
Many CPS operate in dynamic environments with uncertainty, leading to stochastic behaviors. 
For example, autonomous vehicles must meet safety requirements, such as avoiding collisions, despite uncertain road conditions and unpredictable traffic. 
Ensuring safety in such systems requires rigorous verification methods that account for randomness and environmental variability. In recent years, there has been growing interest in safety verification techniques for stochastic dynamical systems, aiming to provide formal guarantees such that they can operate reliably under uncertain conditions~\cite{belta2019formal,lavaei2022automated,yin2024formal}.

Safety verification of stochastic systems is often formulated using the Bellman equations, with the safe probability computed via dynamic programming methods; see, e.g., \cite{abate2008probabilistic}. However, these first-principle methods are in general computationally intractable for high-dimensional systems. An alternative approach involves using abstraction models for verification. For example, \cite{esmaeil2014precise} employs Markov chain abstractions to verify the probabilistic invariance of stochastic systems. A key advantage of abstraction-based methods is their ability to accommodate complex specifications beyond safety~\cite{zhong2023automata} while leveraging well-established algorithms for finite-state systems. Nevertheless, they still suffer from the curse of dimensionality, limiting their scalability in high-dimensional settings.

A more computationally tractable alternative for safety verification is the abstraction-free approach using \emph{barrier certificates}~\cite{prajna2007framework,salamati2022data,nejati2023formal}, which provides Lyapunov-like guarantees on system behavior. Initially, barrier certificates were used to bound the safety probability over an infinite horizon~\cite{prajna2007framework,8263999,10.1145/3126508}. This approach was later extended to finite-horizon safety verification~\cite{steinhardt2012finite,santoyo2021barrier,xue2024finite} since the infinite horizon safety is hard to achieve for most of practical systems. More recently, barrier certificates have been applied to verifying stochastic systems under linear temporal logic (LTL) specifications~\cite{jagtap2018temporal} and to safety verification in  distributionally robust settings, where the exact stochastic distribution is unknown but constrained within an ambiguity set~\cite{schon2024data,chen2025distributionally}.

\textbf{Our Contributions: }
In this work, we address the problem of constructing barrier certificates for formal verification, including both safety and reach-avoid specifications, of stochastic systems over finite time horizons from a dynamic programming (DP) perspective.
First, we revisit the safety verification problem and demonstrate that existing barrier certificate conditions effectively provide a lower bound for the DP solution, which precisely characterizes the safety probability. 
Through this DP perspective, we further reveal that some existing barrier certificate conditions impose unnecessary conservatism over unsafe states. 
Building on this insight, we propose a new type of barrier certificate for safety verification that reduces conservatism by relaxing parameter constraints in prior methods.
We then extend this DP-based construction to reach-avoid specifications, a less explored area in the literature. 
We demonstrate that our approach overcomes the fundamental challenge of extending standard barrier certificate techniques from safety verification to reach-avoid verification. 
Finally, we formulate the search for these new barrier certificates as sum-of-squares (SOS) programs and validate their effectiveness through comparative experiments with existing methods on two illustrative examples. 

\textbf{Related Works: }
The most common approach for constructing finite-time safety barrier certificates relies on the notion of $c$-martingales, which provide an upper bound on the unsafe probability; see, e.g., \cite{steinhardt2012finite,jagtap2018temporal,santoyo2021barrier}. However, this standard framework can be overly conservative in some cases and cannot directly handle reach-avoid verification, as it does not derive a lower bound on the reach-avoid probability. 
Our work addresses these limitations by proposing relaxed conditions for barrier certificate construction compared to existing methods. 
More fundamentally, we establish a unified dynamic programming framework for barrier certificate synthesis, which not only covers both safety and reach-avoid specifications but also enables potential extensions to more complex specifications.
In \cite{xue2024finite}, the authors propose a finite-time reach-avoid barrier certificate using a switched-system-based method. 
Compared to their approach, our formulation achieves a simpler certificate structure and demonstrates better experimental performance. 
We also note that many existing works focus on constructing barrier certificates for infinite-horizon reach-avoid specifications in stochastic systems \cite{anand2022k,vzikelic2023learning,xue2024sufficient,meng2024stochastic,abate2025stochastic,henzinger2025supermartingale}. 
In contrast, our work considers finite-time verification, which requires different analytical tools.

Finally, we would like to clarify that our work is not the  first one attempting to  connect barrier certificates with dynamic programming techniques. 
For instance, in \cite{steinhardt2012finite,laurenti2024unifyingsafetyapproachesstochastic}, the authors also derive safety barrier certificates by showing that their conditions provide an upper bound on the unsafe probability computed via dynamic programming. Compared to these barrier certificates, ours is more general and can reduce to their forms by fixing one parameter.
Moreover, these works primarily reprove barrier certificate results that can also be derived from other perspectives, such as super-martingales.
In contrast, by deeply exploring the connection between barrier certificates and dynamic programming, we identify the conservatism of existing methods over the unsafe region and address it by proposing a new safety barrier certificate.

The rest of the paper is organized as follows. 
We present the necessary preliminaries in Section \ref{sec:prelinimary} and formally formulate the problem in Section~\ref{section:formulation}. 
Next, in Section~\ref{section:revisit}, we revisit existing martingale-based barrier certificates for safety verification from a dynamic programming perspective. 
Building on this new perspective, Section~\ref{sec:dpsafetybound} introduces a novel barrier certificate with relaxed constraints for safety verification and discusses its extension to reach-avoid specifications. 
Section VI demonstrates how to synthesize the proposed barrier certificates using sum-of-squares programming. Finally, we illustrate our method with two case studies in Section~\ref{sec:casestudy} and conclude the work in Section~\ref{sec:con}.

\section{Preliminary}\label{sec:prelinimary}
\textbf{Notations: }
We denote by $\mathbb{R}$, $\mathbb{R}_{\geq 0}$ the set of all real numbers and the set of all non-negative real numbers, respectively. 
Given a Borel space $A$, $\mathcal{B}(A)$ and $\mathcal{P}(A)$ represent its Borel $\sigma$-algebra
and the set of all Borel probability measures on $A$, respectively. 
A transition kernel is a mapping $\mathbf{t}:A\times \mathcal{B}(A)\to [0,1]$ such that $\mathbf{t}(a,\cdot): \mathcal{B}(A) \to [0,1]$ is a Borel probability measure for all $a \in A$ and $\mathbf{t}(\cdot,B): A\to [0,1]$ is measurable for all $B \in \mathcal{B}(A)$.  
The probability measure $\mathbf{t}(a,\cdot)$ is also denoted by $\mathbf{t}(\cdot\mid a)$.
%\mathrm{d}a'

\subsection{Discrete-Time Stochastic Dynamical Systems}
In this work, we consider an autonomous discrete-time stochastic dynamical system, which is a 6-tuple 
\begin{equation}\label{eq:system}
        \mathbb{S}=(X,X_0,W,\mu,f,T),
\end{equation}
where 
\begin{itemize}
    \item $X=\mathbb{R}^n$ is the state space;
    \item $X_0 \subseteq X$ is the set of initial states;
    \item $W \subseteq \mathbb{R}^m$ is the disturbance set;
    \item $\mu \in \mathcal{P}(W)$ is the probability measure of stochastic disturbance;
    \item $f:X \times W \to X$ is continuous function  representing the dynamic of the system;
    \item $T\in \mathbb{N}$ is a positive integer representing the system horizon.
\end{itemize}

Without loss of generality, we do not consider the distribution of the initial state further.  
Instead, we will consider the verification problem \emph{for all} possible initial states.
Given an initial state $x_0\in X_0$, the trajectory of system over horizon $T$ is a sequence
$\mathbf{x}_{0:T}=(x_0,x_1,\dots,x_T)$ such that
\[
 x_{t+1}\sim \mathsf{K}(\cdot\mid x_t), \quad \forall t=0,1,\dots,T-1,
\]
where $\mathsf{K}$ is the transition kernel of system state satisfying that for $B \in \mathcal{B}(X)$ and $x \in X$, $\mathsf{K}(B\mid x)=\mu(\{ w \in W \mid f(x,w) \in B  \})$.
Let $X_i=X$ be the system state set at time $i$ for $i=1,2,\dots, T$. 
We denote by 
\[
\textsf{Pr}_{x_0} \in \mathcal{P}(\{x_0\} \times X_1 \times X_2 \cdots \times X_T)
\]
the unique probability measure induced by the transition kernels $\mathsf{K}$ over the sample space of all state trajectories starting from $x_0$ with horizon $T$~\cite{bertsekas1996stochastic}.

\subsection{Formal Specifications and Their DP Formulations}
In this work, we focus on two types of formal specifications:
\begin{enumerate}
    \item
    the \emph{safety} specification, which requires the system to stay  within a safety set during the entire horizon; and
    \item 
    the \emph{reach-avoid} specification, which requires the system to reach a target set while remaining within a safety set.
\end{enumerate}
Their formal definitions are presented as follows. 

\begin{mydef} 
Let $S \in \mathcal{B}(X)$ be a  Borel set representing  the safe region
and  $G \in \mathcal{B}(X)$ be a  Borel set representing  the target region   with $G \subseteq S$. 
We denote by $\mathbf{x}_{0:T}=(x_0,x_1,\dots,x_T)$ the state trajectory initial from $x_0 \in X_0$.  
The probability of satisfying the safety specification (or safe probability) is defined by
\begin{equation} \label{eq:safetydef}
\mathsf{SA}_{x_0}(S):=\mathsf{Pr}_{x_0}(\{ \mathbf{x}_{0:T} \mid
           \forall t \in [0,T], x_{t} \in S \}), 
\end{equation}   
and the probability of  satisfying the reach-avoid specification (or reach-avoid probability) is defined by
\begin{align} \label{eq:reach-avoiddef} 
&\mathsf{RA}_{x_0}(G,S):=\\
&\mathsf{Pr}_{x_0}(\{  \mathbf{x}_{0:T}  \mid 
\exists t \!\in\! [0,T],   (x_t \!\in\! G  \!\wedge\!  (\forall t' \!\in\! [0,t], x_{t'} \!\in\! S)) \}).\nonumber
\end{align} 
\end{mydef}
According to \cite{abate2008probabilistic,summers2010verification}, the safe probability and the reach-avoid probability, in principle, can be compute precisely  by dynamic programming techniques. 
Formally, let $\mathbf{v}:X\to \mathbb{R}$ be a measurable function on $X$ 
and  we define operator $\mathbb{T}$ on  $\mathbf{v}$ by: for any $x\in X$,
\begin{equation}\label{eq:operator}
\mathbb{T}(\mathbf{v})(x):=\mathbf{1}_{X\setminus S}(x) + \mathbf{1}_{S}(x)  \int_{X} \mathbf{v}(x') \mathsf{K}(\mathrm{d}x'\mid x),  
\end{equation}
where  $\mathbf{1}_{X'}:X\to \{0,1 \}$ is the indicator function such that
\begin{equation}
    \mathbf{1}_{X'}(x)=1
    \ \Leftrightarrow\ x\in X'.
\end{equation}
One can check easily that  operator $\mathbb{T}$ has the following two properties:
\begin{enumerate}[leftmargin=*, labelwidth=1.5em, align=left]
    \item 
    \textbf{monotonicity: } 
    For any two measurable functions 
    $\mathbf{v},\hat{\mathbf{v}}$, we denote by $\mathbf{v} \leq \hat{\mathbf{v}}$ if   $\mathbf{v}(x) \leq \hat{\mathbf{v}}(x),\forall x\in X$. Then
\begin{equation}\label{eq:monoto}
        \mathbf{v}\leq \hat{\mathbf{v}} \implies  \mathbb{T}(\mathbf{v})\leq \mathbb{T}(\hat{\mathbf{v}}).
    \end{equation}
    \item 
    \textbf{linearity in safe region: }  For any  $\rho,\eta \in \mathbb{R}$, 
    let $\rho\mathbf{v}+\eta$ be the function such that $(\rho\mathbf{v}+\eta)(x)=\rho\mathbf{v}(x)+\eta,\forall x\in X$. Then for $x \in S$   we have
\begin{equation}\label{eq:transilationinvariance}
    \mathbb{T}(\rho\mathbf{v}+\eta)(x) = \rho  \mathbb{T}(\mathbf{v})(x) + \eta.
\end{equation}
Note that the linearity only holds for states in safe region. For states in unsafe region, the value of the operator is constant $1$. 
\end{enumerate}

Based on operator $\mathbb{T}$, we define a sequence of value functions  
$\{{\mathbf{v}}_t\}_{t=0,1,\dots,T}$ by:
\begin{align} \label{eq:valuefunctionateachstage}
		\left\{
		\begin{array}{ll }
  \mathbf{v}_T=\mathbf{1}_{X\setminus S}\\
			 \mathbf{v}_t=\mathbb{T}(\mathbf{v}_{t+1}), \forall t=0,1,\dots, T-1
		\end{array}. 
		\right.   
\end{align}
Intuitively, $\mathbf{v}_t(x)$ characterizes the  probability of violating safety within horizon $T-t$ starting from state $x\in X$, 
which is formally stated as follows. 
\begin{mylem}[Lemma 3 in~\cite{abate2008probabilistic}] \label{lem:safetyprobabilitycomputation}
  For each state $x\in X$, it holds that $1-\mathbf{v}_0(x)=\mathsf{SA}_x(S)$.
\end{mylem}
Similarly, let $\mathbf{v}:X\to \mathbb{R}$ be a measurable function on $X$ 
and we define operator $\hat{\mathbb{T}}$ on  $\mathbf{v}$ by: for any $x\in X$,
\begin{equation}\label{eq:reach-avoid-operator}
\hat{\mathbb{T}}(\mathbf{v})(x):=\mathbf{1}_{G}(x) + \mathbf{1}_{S\setminus G}(x)  \int_{X} \mathbf{v}(x') \mathsf{K}(\mathrm{d}x'\mid x).  
\end{equation}
Similar to operator $\mathbb{T}$, one can also verify that operator $\hat{\mathbb{T}}$ is still monotone and linear  in region $S\setminus G$. Based on operator $\hat{\mathbb{T}}$, we  also define a sequence of value functions $\{\hat{\mathbf{v}}_t\}_{t=0,1,\dots,T}$  by: 
\begin{align} \label{eq:reachavoidvaluefunctionateachstage}
		\left\{
		\begin{array}{ll }
  \hat{\mathbf{v}}_T=\mathbf{1}_{G}\\
			 \hat{\mathbf{v}}_t=\hat{\mathbb{T}}(\hat{\mathbf{v}}_{t+1}), \forall t=0,1,\dots, T-1
		\end{array}. 
		\right.   
\end{align}
Intuitively, $\hat{\mathbf{v}}_t(x)$ characterizes the  probability of satisfying the reach-avoid specification  within horizon $T-t$ starting from state $x\in X$. Staring from the initial time instant, we have the following result.  
\begin{mylem}[Lemma 4 in~\cite{summers2010verification}] \label{lem:reachavoidprobabilitycomputation}
  For each state $x\in X$, it holds that $\hat{\mathbf{v}}_0(x)=\mathsf{RA}_x(G,S)$.
\end{mylem}

\section{Problem Formulation}\label{section:formulation}
In this work, we consider the probabilistic verifications of safety and reach-avoid specifications for discrete time stochastic systems. Specifically, given an autonomous stochastic dynamic system, 
our objective is to determine whether the safety (or reach-avoid) probability is higher than a required threshold. The problem is stated formally as follow.
\begin{myprob}\label{problem:verification}
    Given an autonomous discrete-time stochastic dynamical system $\mathbb{S}=(X,X_0,W,\mu,f,T)$, safe region $S \in \mathcal{B}(X)$, target region $G\in\mathcal{B}(X)$, and desired bounds $0\leq \epsilon_{sa},\epsilon_{ra} \leq 1$, determine whether $\forall x \in X_0,  \mathsf{SA}_{x}(S)\geq \epsilon_{sa} $ and $  \mathsf{RA}_{x}(G,S) \geq \epsilon_{ra}$.
\end{myprob}

According to Lemma~\ref{lem:safetyprobabilitycomputation}, 
in order to determine the lower bound for safety   probability, one needs to compute an \emph{upper bound} for the value function $\mathbf{v}_0$. 
On the other hand, 
according to Lemma~\ref{lem:reachavoidprobabilitycomputation}, in order to determine the lower bound for reach-avoid   probability, one needs to compute a \emph{lower bound} for the value function $\hat{\mathbf{v}}_0$. 
Therefore, although the dynamic programming operators $\mathbb{T}$ and $\hat{\mathbb{T}}$ have similar form, we need to approximate the DP operators in different directions.

\section{Revisit Existing Safety Barrier Certificates from Dynamic Programming Perspective}\label{section:revisit}
Explicitly computing the safe probability using dynamic programming is generally a computationally challenging task due to its high complexity.  To mitigate this scalability challenge and compute the safe probability more efficiently, 
for the purpose of safety verification, the notion of \emph{barrier certificates} is widely used in the literature. 
Barrier certificates provide a  \emph{lower bound} on the safe probability, which therefore enables efficient safety verification without requiring exhaustive computation of the exact solution. 

Most existing works on safety verification of stochastic dynamical systems construct barrier certificates based on the notion of $c$-martingales; see, e.g., \cite{steinhardt2012finite, santoyo2021barrier, jagtap2018temporal}. This idea traces back to \cite[Chap. 3]{kushner1967stochastic}. Here, we review the main existing results as follows.

\begin{mylem}\label{lem:martingalebarrier}
A function $\mathbf{v}:X\to \mathbb{R}$ is said to be a \emph{martingale-based safety barrier certificate} (MSBC) if following conditions hold for some $\alpha \geq 1$ and $0\leq \beta <1$:
\begin{enumerate}
    \item 
    $\mathbf{v}\geq \mathbf{1}_{X\setminus S}$;
    \item 
    $\forall x \in S$, $ \int_{X} \mathbf{v}(x') \mathsf{K}(\mathrm{d}x'\mid x) \leq \frac{\mathbf{v}(x)}{\alpha} + \beta$. 
\end{enumerate}
If $\mathbf{v}$ is a MBC, then for $\gamma=\alpha\beta-\alpha+1$, we have \begin{align}\label{eq:martingalebound}
1-\mathsf{SA}_{x}(S) \leq 
		\left\{
		\begin{array}{ll}
			\mathbf{v}(x)(1-\beta)^{T}+ 1-(1-\beta)^{T} &  \text{if } \gamma < 0\\
			\mathbf{v}(x)\alpha^{-T}+\left(\sum_{i=0}^{T-1}\alpha^{-i}\right) \beta      &  \text{if }  \gamma  \geq 0.
		\end{array}
		\right.   
\end{align} 
\end{mylem}

To prove Eq.~\eqref{eq:martingalebound}, existing works adopt a martingale-based approach by showing that the stochastic process $W(x,n)=\alpha^n\mathbf{v}(x)+\alpha\beta(\sum_{i=n}^{T-1}\alpha^i)$ is a super-martingale, i.e., $\mathbb{E}_{x_n}[W(x_{n+1},n+1)]\leq W(x_n,n)$. 
Here, we provide an alternative proof for Lemma~\ref{lem:martingalebarrier}, which connects the martingale-based barrier certificate with the dynamic programming operator. 
The key is to show that the martingale-based barrier certificate, along with its parameters, provides an \emph{upper bound} for the unsafe probability  computed by the dynamic programming operator at each stage.  

Without loss of generality, our proof considers the case where 
$\gamma = \alpha \beta - \alpha + 1 \geq 0$. 
We will later discuss how the case of $\gamma < 0$ can be addressed accordingly. 

\begin{mythm}\label{thm:deriveexisting}
Let $\mathbf{v}$ be a martingale-based barrier certificate  defined in Lemma~\ref{lem:martingalebarrier}
with parameters $\alpha,\beta$  and $\gamma=\alpha \beta -\alpha +1\geq 0$,  and $\{\mathbf{v}_t\}_{t=0,1,\dots,T}$ be value functions defined in Eq.~\eqref{eq:valuefunctionateachstage}. Then, for any $t=0,1,\dots,T$, we have
    \begin{equation}\label{eq:thminductioneq}
  \mathbf{v}_t\leq \alpha^{t-T}\mathbf{v} + (\sum_{i=0}^{T-t-1} \alpha^{-i})\beta,
\end{equation}
which means that $1-\mathsf{SA}_{x}(S) \leq  \mathbf{v}(x)\alpha^{-T}+\left(\sum_{i=0}^{T-1}\alpha^{-i}\right) \beta$. 
\end{mythm}
\begin{proof}
We prove by backward induction from $t=T$. 

\textbf{Induction Basis: }
When $t=T$, according to condition (1) in Lemma~\ref{lem:martingalebarrier}, we have 
\[
\mathbf{v}_T=\mathbf{1}_{X\setminus S}\leq \mathbf{v} = \alpha^{0} \mathbf{v} + (\sum_{i=0}^{-1} \alpha^{-i})\beta.
\]
Therefore, Eq.~\eqref{eq:thminductioneq} holds for $t=T$. 

\textbf{Induction Step: }
Now, let us assume that Eq.~\eqref{eq:thminductioneq} holds for  $t = n + 1 \leq T$, and we aim to prove that Eq.~\eqref{eq:thminductioneq} also holds for $ t = n $. We proceed by considering two distinct cases, depending on whether the state lies within the safe region or not.

Case 1: $x\in S$. \quad
For this case, the correctness of the induction is ensured by  condition (2) in Lemma~\ref{lem:martingalebarrier}, which provides a bound on expectation increasing rate of $\mathbf{v}$ over safe region $S$. 
Formally, for $x \in S$, we have
\begin{align*}
&\mathbf{v}_n(x)=\mathbb{T}(\mathbf{v}_{n+1})(x) \nonumber \\
\overset{(a)}{\leq}& \mathbb{T}\left( \alpha^{n+1-T}\mathbf{v} + (\sum_{i=0}^{T-n-2} \alpha^{-i})\beta \right)(x) \nonumber \\
\overset{(b)}{=} &\alpha^{n+1-T} \mathbb{T}(\mathbf{v})(x) + (\sum_{i=0}^{T-n-2} \alpha^{-i})\beta \nonumber\\
\overset{(c)}{\leq} &\alpha^{n+1-T} \left(\frac{\mathbf{v}(x)}{\alpha}+\beta\right) +(\sum_{i=0}^{T-n-2} \alpha^{-i})\beta \nonumber \\
=&\alpha^{n-T}\mathbf{v}(x)+\alpha^{n+1-T}\beta+(\sum_{i=0}^{T-n-2} \alpha^{-i})\beta \nonumber \\
= & \alpha^{n-T}\mathbf{v}(x)+ (\sum_{i=0}^{T-n-1} \alpha^{-i})\beta. \nonumber\label{eq:safetyregioninduction}
\end{align*}
In the above derivations,  
inequality (a) holds due to the induction hypothesis for $t=n+1\leq T$ and the monotonicity of operator $\mathbb{T}$. 
Equality  (b) comes from the linearity of operator $\mathbb{T}$ in safe region. 
Finally, inequality (c) comes from the fact  that 
$ \mathbb{T}(\mathbf{v})(x)=\int_{X} \mathbf{v}(x') \mathsf{K}(\mathrm{d}x'\mid x) $ 
for $x\in S$ and the barrier condition (2) in Lemma~\ref{lem:martingalebarrier}.

Case 2: $x\in X\setminus S$. \quad
For this case, the correctness of the induction essentially comes from the requirement of $\gamma\geq 0$. 
To see this, for $x\in X\setminus S$, we have
\begin{align}
  &\mathbb{T}(\mathbf{v}_{n+1})(x) =   1  \nonumber \\
  \overset{(a)}{\leq}  & \alpha^{-1} + \beta  \nonumber \\
  =  &\alpha^{-1} \mathbf{v}_{n+1}(x) + \beta  \nonumber \\
   \overset{(b)}{\leq} 
 &\alpha^{-1}\left( \alpha^{n+1-T}\mathbf{v}(x) + (\sum_{i=0}^{T-n-2} \alpha^{-i})\beta \right)+\beta \nonumber \\
 = 
 &\alpha^{n-T}\mathbf{v}(x)+ (\sum_{i=0}^{T-n-1} \alpha^{-i})\beta. \nonumber
\end{align}
In the above derivations,  inequality (a) 
comes from the conditions $\gamma=\alpha\beta-\alpha+1 \geq 0$ and $\alpha>0$. 
Inequality (b) comes from the induction hypothesis on $t=n+1$.  
This completes the proof.
\end{proof}

Note that Theorem~\ref{thm:deriveexisting} only considers the case of $\gamma \geq 0$, as the case of $\gamma < 0$ can be equivalently transformed into the former one without loss of generality. 
Specifically, let us assume that  there is an MSBC $\mathbf{v}$  with $\gamma=\alpha \beta - \alpha +1< 0$. 
Then the function $\mathbf{v}$ is also an MSBC  for new parameters
$\alpha'=(1-\beta)^{-1}$ and $\beta'=\beta$ such that $\alpha'\beta'-\alpha'+1=0$. 
To see this, for any $x \in S$, we have 
\[
\frac{\mathbf{v}(x)}{\alpha'}+\beta'= \mathbf{v}(x)(1-\beta) + \beta > \frac{\mathbf{v}(x)}{\alpha}+\beta, 
\]
where the inequality comes from the condition $\gamma=\alpha\beta - \alpha +1<0$. 
Therefore, the barrier  condition (2) in Lemma~\ref{lem:martingalebarrier} still holds 
when considering new parameters $\alpha'$ and $\beta'$. 
Furthermore,  the probability bound induced by the MSBC $\mathbf{v}$ with parameters $\alpha$ and $\beta$ is
\[
\mathbf{v}(x)(1-\beta)^T+\left(\sum_{i=0}^{T-1}(1-\beta)^i\right) \beta,
\]
which is same as the probability bound induced by the MSBC $\mathbf{v}$ with parameters $\alpha'$ and $\beta'$. 
In other words, the existence of parameters satisfying $\gamma < 0$ always implies the existence of parameters satisfying $\gamma \geq 0$. 
\begin{remark}\label{remark:SBC}
More recently, a different type of safety barrier certificate for stochastic dynamical system is proposed based on an augmented switch system~\cite{xue2024finite}.
Specifically,  a function $\mathbf{v}:X\to \mathbb{R}$ is said to be a 
\emph{switch-system-based safety barrier certificate} (SSBC) if following conditions hold for some $0< \alpha \leq 1$ and $\beta \in \mathbb{R}$:
\begin{enumerate}
        \item $\mathbf{v}\geq \mathbf{1}_{X\setminus S}$;
    \item $\forall x \in S$, $ \int_{X} \mathbf{v}(x') \mathsf{K}(\mathrm{d}x'\mid x) \leq \frac{\mathbf{v}(x)}{\alpha} + \beta$.
\end{enumerate}
If $\mathbf{v}$ is a SSBC, then for $\gamma=\alpha\beta-\alpha+1$, we have
\begin{align}\label{eq:switchbound}
1-\mathsf{SA}_{x}(S) \leq 
		\left\{
		\begin{array}{ll}
			\mathbf{v}(x)(1-\beta)^T+ 1-(1-\beta)^T &  \text{if } \gamma < 0\\
			\mathbf{v}(x)\alpha^{-T}+\left(\sum_{i=0}^{T-1}\alpha^{-i}\right) \beta      &  \text{if }  \gamma  \geq 0.
		\end{array}
		\right.   \nonumber
\end{align} 
The newly proposed SSBC provides a similar probability bound as the MSBC. 
The main difference is that they require a different set of conditions for  parameters $\alpha,\beta \in \mathbb{R}$. 
Our proof for MSBC, from the dynamic programming 
perspective, can also be applied to prove the correctness of SSBC in the same manner. 
\end{remark}

\section{Construction of Barrier Certificates Inspired by Dynamic Programming}\label{sec:dpsafetybound}
Based on the  above discussed new insights between existing barrier certificates and dynamic programming techniques, in this section, we propose a new set of barrier certificates conditions relaxing the existing ones for safety specification. 
Furthermore, we extend our approach to the case of reach-avoid specification by developing a new  barrier certificate for the verification of reach-avoid probability. 

\subsection{Relaxed Barrier Certificates for  Safety Verification}
Recall that, in the proof of Theorem~\ref{thm:deriveexisting}, condition $\gamma\geq 0$ is used in Case 2 $x\in X\setminus S$ 
as a \emph{sufficient condition} to ensure the induction. 
Specifically, it ensures that, for each induction step $t=0,1,\dots,T$, 
\begin{equation}\label{eq:newcondition}
 1\leq \alpha^{t-T}\mathbf{v}(x) + (\sum_{i=0}^{T-t-1} \alpha^{-i})\beta.
\end{equation}
Note that the above equation essentially introduces a new constraint on $\mathbf{v}$, $\alpha$, and $\beta$. 
In other words, if this constraint is enforced during the search for barrier certificates, Theorem~1 remains valid even without the condition $\gamma \geq 0$. 
However, Eq.~\eqref{eq:newcondition} imposes $T+1$ constraints, as it must hold at each induction step. 
It increases computational complexity, particularly when the time horizon $T$ is large.

However, by further analyzing the constraints in Eq.~\eqref{eq:newcondition}, we observe that the value on the right-hand side (RHS) is actually monotonic with respect to $t$ for a fixed $x$. This implies that only two constraints, at the terminal points $t = 0$ and $t = T$, are necessary, rather than enforcing the constraint at every intermediate step.
To see this more clearly, 
let $\eta_t^x$ be the RHS of Eq.~\eqref{eq:newcondition}. 
Then we have
\[
\eta_{t}^x-\eta_{t+1}^x=\alpha^{t-T}(\mathbf{v}(x)-\alpha\mathbf{v}(x)+\alpha\beta).
\]
That is, for each $x\in X\setminus S$,  the value of $\eta_t^x$ is either increasing or decreasing w.r.t. $t$. 
Therefore, to ensure  constraint \eqref{eq:newcondition} true for each time instant,
it suffices to only ensure that  
$\eta^x_T \geq 1$ and $\eta^x_0 \geq 1$ for all $x \in X\setminus S$. 
Furthermore,  
condition (1) in Lemma~\ref{lem:martingalebarrier} has already ensured that $\eta_{T}^x \geq 1,\forall x\in X\setminus S$. 
In summary,  we can improve the existing barrier certificates by dropping  condition $\gamma \geq 0$ and adding a relaxed condition $\eta_{0}^x \geq 1,\forall x \in X \setminus S$. 

Based on the above discussion, we propose a new type of \emph{dynamic-programming-based safety barrier certificates} (DSBC) stated formally as follows.  
\begin{resp}
\begin{mythm}\label{thm:dpbarrier}
    A function $\mathbf{v}:X\to \mathbb{R}$ is said to be a \emph{dynamic programming based safety barrier certificate} (DSBC) if following conditions hold for some $\alpha > 0$ and $\beta \in \mathbb{R}$:
\begin{enumerate}
        \item $\mathbf{v}\geq \mathbf{1}_{X\setminus S}$;
    \item $\forall x \in S$, $ \int_{X} \mathbf{v}(x') \mathsf{K}(\mathrm{d}x'\mid x) \leq \frac{\mathbf{v}(x)}{\alpha} + \beta$;
    \item $\forall x\in X\setminus S, \mathbf{v}(x)\alpha^{-T}+\left(\sum_{i=0}^{T-1}\alpha^{-i}\right) \beta \geq 1$.
\end{enumerate}
If $\mathbf{v}$ is a DSBC, then it holds that \begin{align}\label{eq:dpbound}
1-\mathsf{SA}_{x}(S) \leq \mathbf{v}(x)\alpha^{-T}+\left(\sum_{i=0}^{T-1}\alpha^{-i}\right) \beta.
\end{align} 
\end{mythm}
\end{resp}
\begin{proof}
Let $\eta_t^x=\alpha^{t-T}\mathbf{v}(x) + (\sum_{i=0}^{T-t-1} \alpha^{-i})\beta$. 
For each $x \in X\setminus S$ and $t=0,1,\dots,T$, since
\[
\eta_{t}^x-\eta_{n+1}^x=\alpha^{t-T}(\mathbf{v}(x)-\alpha\mathbf{v}(x)+\alpha\beta),
\]
we know that if $\mathbf{v}(x)-\alpha\mathbf{v}(x)+\alpha\beta \geq 0$, then $\eta_0^x\geq \eta_1^x \dots \geq \eta_T^x$; otherwise $\eta_0^x\leq \eta_1^x \dots \leq \eta_T^x$. 
Therefore,  for $x \in X \setminus S$ and $t=0,1,\dots,T$, we have
\begin{align}\label{eq:improvelarger}
    \eta_t^x \geq& \min\left\{ \eta_0^x=\alpha^{-T}\mathbf{v}(x)+\left(\sum_{i=0}^{T-1}\alpha^{-i}\right) \beta,\eta_T^x=\mathbf{v}(x) \right\}\nonumber\\
    \geq& 1
\end{align}
 where the last inequality comes from conditions 1) and 3). 

Similar to the proof of Theorem~\ref{thm:deriveexisting}, we can show by induction  that 
\begin{equation}\label{eq:improveinduction}
    \mathbf{v}_t\leq \alpha^{t-T}\mathbf{v} + (\sum_{i=0}^{T-t-1} \alpha^{-i})\beta, \quad \forall t=0,1,\dots,T, 
\end{equation}
where $\{\mathbf{v}_t\}_{t=0,1,\dots,T}$ are value functions defined in Eq.~\eqref{eq:valuefunctionateachstage}.
Specifically, 
condition (1) ensures that Eq.~\eqref{eq:improveinduction} holds for $t=T$; 
condition (2) and condition (3) ensure that the induction step holds for $x\in S$ and for $x\in X\setminus S$, respectively.  
\end{proof}

\subsection{Barrier Certificate for Reach-Avoid Verification}
Our approach for safety verification can be further extended to reach-avoid verification.
Compared to the extensive literature on safety verification using barrier certificates, works employing barrier certificates for finite-horizon reach-avoid verification are much rarer. 
The main challenge lies in the difference between the two problems. 
For safety verification, one can adopt the notion of super-martingales to construct an \emph{upper bound} on the value function, which represents the unsafe probability. However, this approach is unsuitable for reach-avoid specifications, as they require deriving a \emph{lower bound} on the corresponding value function.

Nevertheless, from the dynamic programming perspective, reach-avoid barrier certificates can still be constructed similarly to those in Theorem~\ref{thm:dpbarrier}. The result is formally stated in Theorem~\ref{thm:dpreach-avoidbarrier}.  Specifically, condition 1) ensures the correctness of the induction basis, while conditions 2), 3), and 4) guarantee that the induction step holds for states in  $S\setminus G$, $X\setminus S$, and $G$, respectively.  
\begin{resp}
\begin{mythm}\label{thm:dpreach-avoidbarrier}
    A function $\hat{\mathbf{v}}:X\to \mathbb{R}$ is said to be a \emph{reach-avoid barrier certificate} (RABC) if following conditions hold for some $\alpha > 0$ and $\beta \in \mathbb{R}$:
\begin{enumerate}
        \item $\hat{\mathbf{v}}\leq \mathbf{1}_{G}$;
    \item $\forall x \in S\setminus G$, $ \int_{X} \hat{\mathbf{v}}(x') \mathsf{K}(\mathrm{d}x'\mid x) \geq \frac{\hat{\mathbf{v}}(x)}{\alpha} + \beta$;
    \item $\forall x\in X\setminus S, \hat{\mathbf{v}}(x)\alpha^{-T}+\left(\sum_{i=0}^{T-1}\alpha^{-i}\right) \beta \leq 0$;
    \item $\forall x \in G, \hat{\mathbf{v}}(x)\alpha^{-T}+\left(\sum_{i=0}^{T-1}\alpha^{-i}\right) \beta \leq 1$.
\end{enumerate}
If $\hat{\mathbf{v}}$ is a RABC, then it holds that \begin{align}\label{eq:DPRAbound}
\mathsf{RA}_{x}(S,G) \geq \hat{\mathbf{v}}(x)\alpha^{-T}+\left(\sum_{i=0}^{T-1}\alpha^{-i}\right) \beta.
\end{align} 
\end{mythm}
\end{resp}
\begin{proof}
Let $\eta_t^x=\alpha^{t-T}\hat{\mathbf{v}}(x) + (\sum_{i=0}^{T-t-1} \alpha^{-i})\beta$. 
For each $x \in X$ and $t=0,1,\dots,T$, since
\[
\eta_{t}^x-\eta_{n+1}^x=\alpha^{t-T}(\hat{\mathbf{v}}(x)-\alpha\hat{\mathbf{v}}(x)+\alpha\beta),
\]
we know that if $\hat{\mathbf{v}}(x)-\alpha\hat{\mathbf{v}}(x)+\alpha\beta \geq 0$, then $\eta_0^x\geq \eta_1^x \dots \geq \eta_T^x$; otherwise $\eta_0^x\leq \eta_1^x \dots \leq \eta_T^x$. 
Therefore,  for $x \in X \setminus S$ and $t=0,1,\dots,T$, we have
\begin{align}\label{eq:improvelargerunsafe}
    \eta_t^x \leq& \max\left\{ \eta_0^x=\alpha^{-T}\hat{\mathbf{v}}(x)+\left(\sum_{i=0}^{T-1}\alpha^{-i}\right) \beta,\eta_T^x=\hat{\mathbf{v}}(x) \right\}\nonumber\\
    \leq& 0
\end{align}
 where the last inequality comes from conditions 1) and 3). 
 
 Similarly, for $x \in G$ and $t=0,1,\dots,T$, we have
\begin{align}\label{eq:improvelargertarget}
    \eta_t^x \leq& \max\left\{ \eta_0^x=\alpha^{-T}\hat{\mathbf{v}}(x)+\left(\sum_{i=0}^{T-1}\alpha^{-i}\right) \beta,\eta_T^x=\hat{\mathbf{v}}(x) \right\}\nonumber\\
    \leq& 1
\end{align}
  where the last inequality comes from conditions 1) and 4).

Similar to the proof of Theorem~\ref{thm:deriveexisting}, we now show by backward induction that 
\begin{equation}\label{eq:improveinductionreachavoid}
    \hat{\mathbf{v}}_t\geq \alpha^{t-T}\hat{\mathbf{v}} + (\sum_{i=0}^{T-t-1} \alpha^{-i})\beta, \quad \forall t=0,1,\dots,T, 
\end{equation}
where $\{\hat{\mathbf{v}}_t\}_{t=0,1,\dots,T}$ are value functions defined in Eq.~\eqref{eq:reachavoidvaluefunctionateachstage}.

\textbf{Induction Basis: }
When $t=T$, from condition 1), we have 
\[
\hat{\mathbf{v}}_T=\mathbf{1}_{G}\geq \hat{\mathbf{v}} = \alpha^{0} \hat{\mathbf{v}} + (\sum_{i=0}^{-1} \alpha^{-i})\beta.
\] 

\textbf{Induction Step: }
Now, let us assume that Eq.~\eqref{eq:improveinductionreachavoid} holds for  $t = n + 1 \leq T$, and we aim to prove that Eq.~\eqref{eq:improveinductionreachavoid} also holds for $ t = n $. We consider three cases: (a) $x \in S\setminus G$, (b) $x \in G$, and (c) $x \in X\setminus S$.

Case 1: $x \in S\setminus G$.\quad For this case,the correctness of the induction is ensured by condition 2), which provides a bound on expectation decreasing rate of $\hat{\mathbf{v}}$ over region $S\setminus G$. 
Formally, for $x \in S\setminus G$,
\begin{align*}
&\hat{\mathbf{v}}_n(x)=\hat{\mathbb{T}}(\hat{\mathbf{v}}_{n+1})(x) \nonumber \\
\overset{(a)}{\geq}& \hat{\mathbb{T}}\left( \alpha^{n+1-T}\hat{\mathbf{v}} + (\sum_{i=0}^{T-n-2} \alpha^{-i})\beta \right)(x) \nonumber \\
\overset{(b)}{=} &\alpha^{n+1-T} \hat{\mathbb{T}}(\hat{\mathbf{v}})(x) + (\sum_{i=0}^{T-n-2} \alpha^{-i})\beta \nonumber\\
\overset{(c)}{\geq} &\alpha^{n+1-T} \left(\frac{\hat{\mathbf{v}}(x)}{\alpha}+\beta\right) +(\sum_{i=0}^{T-n-2} \alpha^{-i})\beta \nonumber \\
=&\alpha^{n-T}\hat{\mathbf{v}}(x)+\alpha^{n+1-T}\beta+(\sum_{i=0}^{T-n-2} \alpha^{-i})\beta \nonumber \\
= & \alpha^{n-T}\hat{\mathbf{v}}(x)+ (\sum_{i=0}^{T-n-1} \alpha^{-i})\beta. \nonumber\label{eq:safetyregioninduction}
\end{align*}
In the above derivations,  
inequality (a) holds due to the induction hypothesis for $t=n+1\leq T$ and the monotonicity of operator $\hat{\mathbb{T}}$. 
Equality  (b) comes from the linearity of operator $\hat{\mathbb{T}}$ in region $S\setminus G$. 
Finally, inequality (c) comes from the fact  that 
$ \hat{\mathbb{T}}(\hat{\mathbf{v}})(x)=\int_{X} \hat{\mathbf{v}}(x') \mathsf{K}(\mathrm{d}x'\mid x) $ 
for $x\in S\setminus G$ and the condition 2).

Case 2: $x \in G$. \quad For $x \in G$, $t=1,2,\dots,T$, we have
\[
\hat{\mathbf{v}}_t(x)=1\geq \alpha^{t-T}\hat{\mathbf{v}}(x) + (\sum_{i=0}^{T-t-1} \alpha^{-i})\beta
\]
where the inequality comes from \eqref{eq:improvelargertarget}.

Case 3: $x \in X\setminus S$. \quad For $x \in X\setminus S$, $t=1,2,\dots,T$, we have
\[
\hat{\mathbf{v}}_t(x)=0\geq \alpha^{t-T}\hat{\mathbf{v}}(x) + (\sum_{i=0}^{T-t-1} \alpha^{-i})\beta
\]
where the inequality comes from \eqref{eq:improvelargerunsafe}. This completes the proof.
\end{proof}

\begin{remark} 
Our dynamic programming approach to construct barrier certificates offers several potential advantages. First, it can be naturally extended to any specification whose satisfaction probability can be efficiently computed via dynamic programming. For example, one can use a set of conditions to bound the initial dynamic programming value function and each employment of the dynamic programming operator, as those in Theorem~\ref{thm:dpbarrier} and Theorem~\ref{thm:dpreach-avoidbarrier}. Second, beyond lower-bounding satisfaction probabilities, our method can also derive upper bounds for safety and reach-avoid probabilities in a similar framework. This dual-bound capability enables tighter estimation of the true satisfaction probability, which may further enhance verification precision.
\end{remark}

\section{Computation of Barrier Certificate by Sum of Square Programs}\label{subsection:SOSsomputation}
So far, we have proposed two barrier certificates: one for safety specification by removing the parameter condition of existing method and adding only single condition, the other for reach-avoid specification by a similar construction manner.
Now we show that the problems of searching the proposed barrier certificates can be solved by sum of square (SOS) programs~\cite{sturm1999using, prajna2002introducing}.

Given $ z \in \mathbb{R}^d $,  let $ \mathbb{R}[z] $ denote the set of all polynomials over $ z $. 
The set of SOS polynomials over $z$ is 
\begin{equation}\label{def:polynomial}
    \Sigma[z] := \left\{ s(z) \in \mathbb{R}[z] \mid s(z)= \sum_{i=1}^m g_i(z)^2, g_i(z) \in \mathbb{R}[z] \right\}. \nonumber
\end{equation}
In order to apply SOS optimization techniques, 
we need to make the following further  assumptions.
\begin{assumption}\label{assumption:polynomial}
The following assumptions hold for system $\mathbb{S}=(X,X_0,W,\mu,f,T)$, safety set $S \in \mathcal{B}(X)$, and target set $G \in \mathcal{B}(X)$: 
\begin{itemize}
    \item[A1] 
    The system dynamic $f$ is a polynomial over $X\times W$;
    \vspace{3pt}
    \item[A2] 
    Each set $\star \in \{X_0,S,X \setminus S\}$ is defined as the super-level set of polynomial $\mathbf{s}_{\star}$, 
    i.e., $\star=\{ x \in \mathbb{R}^n \mid \mathbf{s}_{\star}(x) \geq 0 \}$;
    \vspace{3pt}
    \item[A3] Each set $\star \in \{X_0, G, S\setminus G,X \setminus S,X\setminus G\}$ is defined as the super-level set of polynomial $\mathbf{s}_{\star}$.
\end{itemize}
\end{assumption}
The assumptions above are commonly adopted in the literature when searching for barrier certificates using SOS programs.
Under Assumption~\ref{assumption:polynomial}, we propose the SOS program~\eqref{eq:objSOS}-\eqref{eq:fifthSOS} to search for the DSBC. 
For each search, the value of $\alpha$ must be given in advance to ensure the linearity of the objective function.
The variables $\beta$ and $\delta$ are program variables, and $\Tilde{\mathbf{v}}$ represents the computed DSBC. 
Specifically, constraints \eqref{eq:firstSOS} and \eqref{eq:secondSOS} ensure that $\Tilde{\mathbf{v}} \geq \mathbf{1}_{X\setminus S}$, while constraint \eqref{eq:thirdSOS} guarantees that $\Tilde{\mathbf{v}}(x) \leq \delta$ for all $x \in X_0$. Furthermore, constraints \eqref{eq:forthSOS} and \eqref{eq:fifthSOS} ensure that conditions 2) and 3) in Theorem~2 are satisfied, respectively. Finally, the objective \eqref{eq:objSOS} minimizes the unsafe probability over the initial state set.
 \begin{tcolorboxenv}[t]
\begin{center}
        \tcbset{title=\quad SOS Program for Searching DSBC}
    \begin{tcolorbox} 
        \vspace{-15pt}
\begin{align}
 \min_{\delta, \beta \in \mathbb{R}} \quad \delta\alpha^{-T}+\left(\sum_{i=0}^{T-1}\alpha^{-i}\right)\beta &  \label{eq:objSOS} \\
 \Tilde{\mathbf{v}}(x)-\xi_S(x) \mathbf{s}_{S}(x) & \in \Sigma[x]  \label{eq:firstSOS} \\
 \Tilde{\mathbf{v}}(x)-\xi_{X \setminus S}(x) \mathbf{s}_{X \setminus S}(x)-1 & \in \Sigma[x] \label{eq:secondSOS} \\
 - \Tilde{\mathbf{v}}(x)-\xi_{X_0}(x) \mathbf{s}_{X_0}(x) +\delta & \in \Sigma[x] \label{eq:thirdSOS} \\
-\int_{X} \Tilde{\mathbf{v}}(x') \mathsf{K}(\mathrm{d}x'\mid x) +\Tilde{\mathbf{v}}(x)/\alpha+\beta &    \nonumber   \\ 
 -\xi_{S}(x) \mathbf{s}_{S}(x)  &\in \Sigma[x] \label{eq:forthSOS} \\
  \Tilde{\mathbf{v}}(x)\alpha^{-T}+\left(\sum_{i=0}^{T-1}\alpha^{-i}\right)\beta-1 &\nonumber \\
  - \xi_{X \setminus S}(x) \mathbf{s}_{X \setminus S}(x)   &\in \Sigma[x]    \label{eq:fifthSOS}
\end{align} 
 \end{tcolorbox}
\end{center}
\end{tcolorboxenv} 
\begin{mythm}\label{thm:barrier}
Suppose that A1 and A2 in Assumption \ref{assumption:polynomial} hold and $\alpha >0 $ is given.
Let $\Tilde{\mathbf{v}}(x) \in \mathbb{R}[x],\xi_{X_0}(x),\xi_{S}(x)$, $\xi_{X\setminus S}(x) \in \Sigma[x]$, $\delta$ and $\beta$ be the solutions to SOS program \eqref{eq:objSOS}-\eqref{eq:fifthSOS}. 
Then we have
\begin{equation}\label{eq:barrierresult}
\forall x \in X_0: 1- \mathsf{SA}_{x}(S)\leq  \delta\alpha^{-T}+\left(\sum_{i=0}^{T-1}\alpha^{-i}\right) \beta .
\end{equation}
\end{mythm}
\begin{proof}
First, we prove that the computed polynomial $\Tilde{\mathbf{v}}(x) \in \mathbb{R}[x]$ is a DSBC as defined in Theorem~\ref{thm:dpbarrier}. 
We show that the satisfaction of \eqref{eq:firstSOS} implies satisfaction of $\Tilde{\mathbf{v}}(x) \geq 0, \forall x \in S$. Specifically, if \eqref{eq:firstSOS} holds, then for any $x \in S$, we have
 \[
 \begin{aligned}
   \Tilde{\mathbf{v}}(x) \geq  \Tilde{\mathbf{v}}(x)-\xi_{S}(x) \mathbf{s}_{S}(x)\geq 0.
 \end{aligned}
 \]
 The first inequality comes from the fact that $\mathbf{s}_{S}(x)\xi_{S}(x) \geq 0$ holds for all $x \in S$. The last equality holds since constraint \eqref{eq:firstSOS} ensures that   expression $\Tilde{\mathbf{v}}(x)-\xi_{S}(x) \mathbf{s}_{S}(x)$ is an SOS polynomial. 
 Similarly, we can prove that if \eqref{eq:secondSOS}; \eqref{eq:thirdSOS}; \eqref{eq:forthSOS}; \eqref{eq:fifthSOS} hold, then $\Tilde{\mathbf{v}}(x) \geq 1, \forall x \in X\setminus S$; $\Tilde{\mathbf{v}}(x) \leq \delta, \forall x \in X_0$; $ \int_{X} \Tilde{\mathbf{v}}(x') \mathsf{K}(\mathrm{d}x'\mid x) \leq \frac{\Tilde{\mathbf{v}}(x)}{\alpha} + \beta, \forall x \in S$; $\Tilde{\mathbf{v}}(x)\alpha^{-T}+\left(\sum_{i=0}^{T-1}\alpha^{-i}\right) \beta\geq 1, \forall x \in X\setminus S$ hold, respectively. 
 Then by definition, we know that the computed polynomial $\Tilde{\mathbf{v}}$ is indeed a DSBC. 
 Finally, by Eq.~\eqref{eq:dpbound} and $\Tilde{\mathbf{v}}(x) \leq \delta, \forall x \in X_0$, we know that Eq.~\eqref{eq:barrierresult} holds. This completes the proof.
\end{proof}
 \begin{tcolorboxenv}[t]
\begin{center}
        \tcbset{title=\quad SOS Program for Searching RABC}
    \begin{tcolorbox} 
        \vspace{-15pt}
\begin{align}
 \max_{\delta, \beta \in \mathbb{R}} \quad \delta\alpha^{-T}+\left(\sum_{i=0}^{T-1}\alpha^{-i}\right)\beta &  \label{eq:objSOSRA} \\
 -\Tilde{\mathbf{v}}(x)-\xi_{X\setminus G}(x) \mathbf{s}_{X\setminus G}(x) & \in \Sigma[x]  \label{eq:firstSOSRA} \\
 -\Tilde{\mathbf{v}}(x)-\xi_{G}(x) \mathbf{s}_{G}(x)+1 & \in \Sigma[x] \label{eq:secondSOSRA} \\
 \Tilde{\mathbf{v}}(x)-\xi_{X_0}(x) \mathbf{s}_{X_0}(x) -\delta & \in \Sigma[x] \label{eq:thirdSOSRA} \\
\int_{X} \Tilde{\mathbf{v}}(x') \mathsf{K}(\mathrm{d}x'\mid x) -\Tilde{\mathbf{v}}(x)/\alpha-\beta &    \nonumber   \\ 
 -\xi_{S\setminus G}(x) \mathbf{s}_{S\setminus G}(x)  &\in \Sigma[x] \label{eq:forthSOSRA} \\
  -\Tilde{\mathbf{v}}(x)\alpha^{-T}-\left(\sum_{i=0}^{T-1}\alpha^{-i}\right)\beta+1 &\nonumber \\
  - \xi_{G}(x) \mathbf{s}_{G}(x)   &\in \Sigma[x]    \label{eq:fifthSOSRA} \\
    -\Tilde{\mathbf{v}}(x)\alpha^{-T}-\left(\sum_{i=0}^{T-1}\alpha^{-i}\right)\beta &\nonumber \\
  - \xi_{X \setminus S}(x) \mathbf{s}_{X \setminus S}(x)   &\in \Sigma[x]    \label{eq:sixthSOSRA}
\end{align} 
 \end{tcolorbox}
\end{center}
\end{tcolorboxenv}

For reach-avoid verifications, we also propose SOS program \eqref{eq:objSOSRA}-\eqref{eq:sixthSOSRA} to search for the RABC.
Similar to the program for DSBC,  $\Tilde{\mathbf{v}}$ is the computed RABC. 
Moreover, the constraints \eqref{eq:firstSOSRA}-\eqref{eq:sixthSOSRA} ensure $\Tilde{\mathbf{v}}$ to satisfy the conditions of RABC in Theorem~\ref{thm:dpreach-avoidbarrier}. 
Then we have the following result. 

\begin{mythm}\label{thm:barrierRA}
Suppose that A1 and A3 in Assumption \ref{assumption:polynomial} hold and $\alpha >0 $ is given.
Let $\Tilde{\mathbf{v}}(x) \in \mathbb{R}[x],\xi_{X_0}(x),\xi_{G}(x),\xi_{S\setminus G}(x),\xi_{X\setminus S}(x),$ $\xi_{X\setminus G}(x) \in \Sigma[x]$, $\delta$ and $\beta$ be the solutions to SOS program \eqref{eq:objSOSRA}-\eqref{eq:sixthSOSRA}. 
Then we have
\begin{equation}\label{eq:barrierresultRA}
\forall x \in X_0:  \mathsf{RA}_{x}(G,S)\geq  \delta\alpha^{-T}+\left(\sum_{i=0}^{T-1}\alpha^{-i}\right) \beta .
\end{equation}
\end{mythm}
\begin{proof}
First, we prove that the computed polynomial $\Tilde{\mathbf{v}}(x) \in \mathbb{R}[x]$ is a RABC as defined in Theorem~\ref{thm:dpreach-avoidbarrier}.
Specifically, we can prove similar with Theorem~\ref{thm:barrier} that if (a) \eqref{eq:firstSOSRA} and \eqref{eq:secondSOSRA}; (b) \eqref{eq:forthSOSRA}; (c) \eqref{eq:fifthSOSRA}; (d) \eqref{eq:sixthSOSRA} hold, then $\Tilde{\mathbf{v}}$ satisfies the condition of (a) 1); (b) 2); (c) 4); (d) 3) in Theorem~\ref{thm:dpreach-avoidbarrier}, respectively.
 Then by definition, we know that the computed polynomial $\Tilde{\mathbf{v}}$ is indeed a RABC. 
 Moreover, we can get $\Tilde{\mathbf{v}}(x) \geq \delta, \forall x \in X_0$ from \eqref{eq:thirdSOSRA}.
 Finally, by Eq.~\eqref{eq:DPRAbound}, we know that Eq.~\eqref{eq:barrierresultRA} holds. This completes the proof.
\end{proof}
In the proposed SOS programs, the integral term $\int_{X} \Tilde{\mathbf{v}}(x') \mathsf{K}(\mathrm{d}x'\mid x)$ in Eq.~\eqref{eq:forthSOS} involves computing the $n$-th moment of a random variable. This is a computationally challenging task for an arbitrary distribution. However, for many widely used random variables representing system noise, such as those following a uniform distribution or a zero-mean normal distribution, closed-form expressions can be derived. 
For example, the expected value of the $n$-th moment of a uniform random variable $z$ over the interval $[a, b]$ is given by:
\begin{equation}\label{eq:uniform}
    \mathbb{E}[z^n] = \frac{b^{n+1}-a^{n+1}}{(n+1)(b-a)}.
\end{equation}
When the random variable $z$ follows a zero-mean normal distribution with standard deviation $\sigma$, we have: 
\begin{align}\label{eq:normal}
\mathbb{E}[z^n] = 
		\left\{
		\begin{array}{ll}
			0 &  \text{if } n \text{ is odd}\\
		1\cdot 3\cdots (n-1) \sigma^n      &  \text{if } n \text{ is even}
		\end{array}.
		\right.   
\end{align}

\section{Case Study} \label{sec:casestudy}

\begin{table*}[ht]
  \centering
\caption{Upper bound of unsafe probability $1-\textsf{SA}_{x_0}(S)$ by different barrier certificates under different $\alpha$ in Example 1}\label{tab:example 1}
  \begin{threeparttable}[b]
     \begin{tabular}{cccccccccccc}
      \toprule
      $\alpha$ value & $0.99$ & $0.992$ & $0.994$ & $0.996$ & $0.998$ & $1$ & $1.002$ & $1.004$ & $1.006$ & $1.008$ & $1.01 $ \\ 
       \midrule
     DBC \textbf{(Ours)} & $0.9681$ & $0.8759$ & $0.793$ & $0.7178$ & $0.65$ & $0.5896$ & $0.5891$ & $0.5895$ & $0.59$ & $0.5906$ & $0.5915$\\
     MBC (Lemma~\ref{lem:martingalebarrier})  & $\backslash$ & $\backslash$ & $\backslash$ & $\backslash$ & $\backslash$  & $0.5903$ & $0.6225$ & $0.6565$ & $0.6881$ & $0.7167$ & $0.7428$\\
     SBC (Remark~\ref{remark:SBC})  & $0.9681$ & $0.8759$ & $0.793$ & $0.7177$ & $0.65$  & $0.5907$ & $\backslash$ & $\backslash$ & $\backslash$ & $\backslash$ & $\backslash$\\
      \bottomrule
    \end{tabular}
  \end{threeparttable}
\end{table*}
\begin{table*}[ht]
  \centering
\caption{Upper bound of unsafe probability $1-\textsf{SA}_{x_0}(S)$ by different barrier certificates under different $\alpha$ in Example 2}\label{tab:example 2}
  \begin{threeparttable}[b]
     \begin{tabular}{cccccccccccc}
      \toprule
      $\alpha$ value & $0.9$ & $0.92$ & $0.94$ & $0.96$ & $0.98$ & $1$ & $1.02$ & $1.04$ & $1.06$ & $1.08$ & $1.1 $ \\
       \midrule
     DBC \textbf{(Ours)} & $0.249$ & $0.2176$ & $0.1984$ & $0.19$ & $0.1929$ & $0.21$ & $0.5$ & $0.69$ & $0.8129$ & $0.8936$ & $0.9465$  \\
     MBC (Lemma~\ref{lem:martingalebarrier})  & $\backslash$ & $\backslash$ & $\backslash$ & $\backslash$ & $\backslash$  & $0.21$ & $0.5$ & $0.69$ & $0.8129$ & $0.8936$ & $0.9465$\\
     SBC (Remark~\ref{remark:SBC})  & $0.2491$ & $0.2176$ & $0.1984$ & $0.19$ & $0.1929$  & $0.21$ & $\backslash$ & $\backslash$ & $\backslash$ & $\backslash$ & $\backslash$\\
      \bottomrule
    \end{tabular}
  \end{threeparttable}
\end{table*}
 \begin{table*}[ht]
  \centering
\caption{Value of $\gamma=\alpha\beta-\alpha+1$ for DBCs under different values of $\alpha$ in two examples. 
\\The value index represents the corresponding column  of $\alpha$ values in Table \ref{tab:example 1} and Table \ref{tab:example 2}}\label{tab:gammavalue}
  \begin{threeparttable}[b]
     \begin{tabular}{cccccccccccc}
      \toprule
    Value index & 1 & 2 & 3 & 4 & 5 & 6 & 7 & 8 & 9 & 10 & 11 \\ 
      \midrule
     Example 1 & $0.01$ & $0.008$ & $0.006$ & $0.004$ & $0.002$ & $8\times 10^{-6}$ & $-0.002$ & $-0.004$ & $-0.006$ & $-0.008$ & $-0.01$ \\
     Example 2 & $0.1031$ & $0.0837$ & $0.0644$ & $0.0456$ & $0.0273$ & $0.01$ & $0.0102$ & $0.0104$ & $0.0106$ & $0.0108$ & $0.011$ \\
      \bottomrule
    \end{tabular}
  \end{threeparttable}
\end{table*}
 \begin{table*}[ht]
  \centering
\caption{Lower bound of reach-avoid probability $\textsf{RA}_{x_0}(G,S)$ by different method under different degrees in Example 1}\label{tab:reachavoidcompa}
  \begin{threeparttable}[b]
     \begin{tabular}{cccccccccccc}
      \toprule
    Degree of polynomials & 2 & 4 & 6 & 8 & 10 & 12 & 14  \\ 
      \midrule
     RABC (\textbf{Ours}) & $0.2027$ & $0.3689$ & $0.3995$ & $0.4136$ & $0.4953$ & $0.5547$ & $0.6080$ \\
     (13) in \cite{xue2024finite} & $0.1591$ & $0.2824$ & $0.3453$ & $0.3669$ & $0.4606$ & $0.5218$ & $0.5732$  \\
      \bottomrule
    \end{tabular}
  \end{threeparttable}
\end{table*}

In this section, 
we illustrate our  proposed new barrier certificates 
and compare them with existing methods using two examples adapted from \cite{xue2024finite}.

\textbf{Example 1:} Consider the following discrete-time system
\[
x(t+1)=f(x(t),w(t))=(-0.5+w(t))x(t),
\]
where $x(t) \!\in\! X\!=\!\mathbb{R}$, $x_0\!=\!-0.9$, $w(t)$ follows the uniform distribution over $W=[-1,1]$ and time horizon $T=50$. The safety set is $S=[-1,1]$ and the exact unsafe probability  obtained by exhaustive computation is $1-\mathsf{SA}_{x_0}(S)=0.2321$.
The target set is $G=[-0.6,0.6]$ and the exact  reach-avoid probability  obtained by exhaustive computation  is $\mathsf{RA}_{x_0}(G,S)=0.7708$.

\textbf{Example 2:} 
Consider the  following discrete-time system 
\[
x(t+1)=f(x(t),w(t))=x(t)+w(t)
\]
where $x(t)\!\in\! X \!=\!\mathbb{R}$, $X_0=[-0.1,0.1]$, $w(t)$ follows the normal distribution over $W=\mathbb{R}$ with  zero mean and $0.1$ standard variation. The time horizon is $T=20$ and the safety set is $S=[-1,1]$. 

\textbf{Experiments for Safety Verification: }
For both Examples~1 and 2, we use the SOS program \eqref{eq:objSOS}-\eqref{eq:fifthSOS} to search for DSBCs under different values of $\alpha > 0$. For comparison, we also search for the existing MSBC defined in Lemma~\ref{lem:martingalebarrier} and the SSBC discussed in Remark~\ref{remark:SBC} using SOS programs. As discussed, it is without loss of generality to consider only the case $\gamma \geq 0$ for MSBC and SSBC. Therefore, compared to the SOS program \eqref{eq:objSOS}-\eqref{eq:fifthSOS}, the SOS programs for searching MSBC and SSBC will omit the constraint \eqref{eq:fifthSOS} and include the constraint $\gamma = \alpha\beta - \alpha + 1 \geq 0$. The degree of polynomials in the SOS programs is set to $6$ and $4$ for Example $1$ and $2$, respectively.

Table~\ref{tab:example 1} and \ref{tab:example 2} present the experimental results for Examples $1$ and $2$, respectively. 
According to the tables, our DSBC generates valid bounds for all values of $\alpha$, whereas the MSBC and SSBC are only valid for certain values of $\alpha$.
Additionally, in Example $1$, the DSBC provides tighter bounds for the unsafe probability compared to the MSBC when $\alpha \geq 1$. 
Therefore, compared to existing approaches, our proposed DSBC not only increases the success rate for finding a valid barrier certificate but also provides more precise probability bounds.
On the other hand, the proposed DSBC generates probability bounds that are identical to those produced by the other methods in Example $2$, which means that the improvement of our barrier certificates does not always happen when existing bounds are already quite tight.

As illustrated in Table~\ref{tab:gammavalue}, if $\gamma=\alpha\beta-\alpha+1<0$ for computed DSBC, our DSBC will have better result than other methods. Specifically, if the computed DBC satisfies $\gamma < 0$, it means that the condition $\gamma\geq 0$ in existing barrier certificate is an active constraint in SOS optimization program. Since DSBC replaces this constraint by a relaxed one, it can generate better result in such case. 
Moreover, we find that the performance of MSBC and SSBC are dependent on the considered systems. The MSBC has better performance in Example 1 while the situation is opposite in Example 2. Therefore, since the proposed DSBC unifies two existing barrier certificates, we can find the best probability bound by only considering the DSBC. 

\textbf{Experiments for Reach-Avoid Specification: } 
Here, in order to verify reach-avoid probability, we use the SOS program \eqref{eq:objSOSRA}-\eqref{eq:sixthSOSRA} to search for the RABC only for Example~1. 
As discussed before, the most popular $c$-martingale approach can only provides upper bound of the reach-avoid probability for discrete time stochastic system with finite horizon, which is unsuitable for reach-avoid verification.
The only related method is \cite{xue2024finite}, which also constructed a set of conditions to provide a lower bound for reach-avoid probability by SOS program.  This method requires two parameters $\alpha>1$ and $\beta \geq 1-\alpha$ to be given ahead for the linearity of the objective function of optimization problem. On the contrary, our RABC only requires the value of $\alpha$ to be given and can find optimal $\beta$ directly.

For the sake of comparison, we compute the proposed RABC with $\alpha=1.06$, and (13) in \cite{xue2024finite} with $\alpha=1.06$ and $\beta=0$ using different degrees of polynomials. The experimental results are shown in Table~\ref{tab:reachavoidcompa}. Compared with the method in \cite{xue2024finite}, our proposed RABC  can generate better bound for reach-avoid probability. Moreover, as the degree of polynomials in SOS programs increases, both methods can provide larger lower bound of reach-avoid probability and the performance gap between methods becomes closer.

\section{Conclusion}\label{sec:con}
In this work, we revisited the formal verification problem for stochastic dynamical systems using barrier certificates. 
We first provided an alternative proof for the existing martingale-based barrier certificates for safety specifications from the perspective of dynamic programming. Building on this new perspective, we further proposed a new type of barrier certificate that not only unifies existing ones but also provides more relaxed constraints. 
We showed such construction can also be applied to the reach-avoid verification problem.  
We demonstrated how the new barrier certificates can be searched for using SOS programs. 
Our results offer new insights into barrier certificates for stochastic dynamical systems, and we believe there are several promising future directions stemming from this perspective. For instance, an immediate direction is to extend our approach to more general specifications, such as  temporal logic specifications. Additionally, we aim to explore how our approach can be applied to synthesizing formal control strategies \cite{jagtap2020formal,zhong2025secure}, beyond verification. 

\bibliographystyle{ieeetr}
\bibliography{main}

\end{document}